\documentclass[12pt]{llncs}

\usepackage[english]{babel}

\usepackage{lmodern}
\usepackage[T1]{fontenc}
\usepackage[utf8]{inputenc}

\usepackage{epic}
\usepackage{eepic}
\usepackage{amsmath}
\usepackage{amssymb}

\usepackage{color}
\usepackage{latexsym}
\usepackage{graphicx}
\usepackage{exscale}
\usepackage{epsfig}
\usepackage{rotate}

\newcommand{\opt}{\mbox{\sc Opt}}
\newcommand{\np}{\mbox{\sf NP}}

\newcommand{\hide}[1]{}

\usepackage{xspace}

\newcommand{\displayqed}{\eqno{\hbox{\hfill\qed}}}

\title{The Max-Distance Network Creation Game \\on General Host Graphs\thanks{A preliminary version of this paper was presented at the
\emph{8th International Workshop on Internet \& Network Economics (WINE'12)}, Liverpool, UK, December 9-12, 2012, and appeared in
Vol. 7695 of Lecture Notes in Computer Science, Springer, 393--406. DOI: \mbox{http://dx.doi.org/10.1007/978-3-642-35311-6\_29} -- This work was
partially supported by the Research Grant PRIN 2010 ``ARS
TechnoMedia'', funded by the Italian Ministry of Education, University, and Research.
}}
\author{Davide Bil\`o\inst{1} \and
Luciano Gual\`a\inst{2} \and Stefano Leucci\inst{2} \and Guido Proietti\inst{3,4}
\institute{Dipartimento di Scienze Umanistiche e Sociali,
Università di Sassari, Italy \and Dipartimento di Ingegneria dell'Impresa,
Università di Roma ``Tor Vergata'', Italy \and Dipartimento di
Ingegneria e Scienze dell'Informazione e Matematica, \\Università degli Studi dell'Aquila, Italy  \and Istituto di Analisi dei Sistemi
  ed Informatica,
  CNR, Rome, Italy \\
E-mail: {\tt davide.bilo@uniss.it; guala@mat.uniroma2.it; stefano.leucci@graduate.univaq.it;
guido.proietti@univaq.it} }}

\begin{document}

\pagestyle{plain}

\maketitle

\begin{abstract}
{\fontsize{10}{12} \selectfont
In this paper we study a generalization of the classic \emph{network creation game} in the scenario in which the $n$ players sit on a given arbitrary \emph{host graph}, which constrains the set of edges a player can activate at a cost of $\alpha \geq 0$ each. This finds its motivations in the physical limitations one can have in constructing links in practice, and it has been studied in the past only when the routing cost component of a player is given by the sum of distances to all the other nodes. Here, we focus on another popular routing cost, namely that which takes into account for each player its \emph{maximum} distance to any other player. For this version of the game, we first analyze some of its computational and dynamic aspects, and then we address the problem of understanding the structure of associated pure Nash equilibria. In this respect, we show that the corresponding price of anarchy (PoA) is fairly bad, even for several basic classes of host graphs. More precisely, we first exhibit a lower bound of $\Omega (\sqrt{ n / (1+\alpha)})$ for any $\alpha = o(n)$. Notice that this implies a counter-intuitive lower bound of $\Omega(\sqrt{n})$ for very small values of $\alpha$ (i.e., edges can be activated almost for free). Then, we show that when the host graph is restricted to be either $k$-regular (for any constant $k \geq 3$), or a 2-dimensional grid, the PoA is still $\Omega(1+\min\{\alpha, \frac{n}{\alpha}\})$, which is proven to be tight for $\alpha=\Omega(\sqrt{n})$. On the positive side, if $\alpha \geq n$, we show the PoA is $O(1)$. Finally, in the case in which the host graph is very sparse (i.e., $|E(H)|=n-1+k$, with $k=O(1)$), we prove that the PoA is $O(1)$, for any $\alpha$.

\bigskip


}
\end{abstract}

\setcounter{footnote}{0}
\section{Introduction}
In a \emph{network creation game} (NCG), we are given $n$ players identified as the nodes of a graph, and each player attempts to connect itself to all the other players. In such a decentralized process, each player aims to selfishly optimize a certain \emph{routing} cost towards the other players. Thus, its action consists of choosing a suitable subset of players, which are then made adjacent through the activation of the corresponding set of incident edges. Unavoidably, activating a link incurs a cost to the player, and so the overall \emph{building} cost should be strategically balanced with the aforementioned routing cost.

Due to their generality, it is in clear evidence that NCGs can model very different practical situations, depending on how all the build-up ingredients are mixed.
In the very classic formulation of the game \cite{JW96}, each player has no limitations in choosing a subset of adjacent players, its routing cost is a function of the sum of distances to all the other players (i.e., the so called \emph{sum cost}), and activating a link has a fixed cost $\alpha \ge 0$. Not surprisingly, this model was devised by the economists, which were mainly interested in understanding whether the attainment of an equilibrium status (i.e., a status in which
players are not willing to move from) for a mutual-relationships social system is compatible with the behavior of the players, which tend to establish selfishly their personal contacts.

With the recent advent of the algorithmic game theory, the interest on NCGs has been reawakened. This is especially due to the fact that NCGs are fit to model the decentralized construction of \emph{communication} networks, in which the constituting components (e.g., routers and links) are activated and maintained by different owners, as in the Internet. According to its performance measurement philosophy, computer scientists put a new special emphasis on the challenge of understanding how the social utility for the (very large) system as a whole is affected by the selfish behavior of the players. This trend originated from the paper of Fabrikant \emph{et al.} \cite{FLM03}, and was then followed by a sequel of papers, as detailed in the following.

\paragraph{Previous work.}
As said before, the canonical form of a NCG, also known as \textsc{Sum}NCG, is as follows: We are
given a set of $n$ players, say $V$, where the
strategy space of player $v \in V$ is the power set $2^{V \setminus\{v\}}$. Given
a combination of strategies $\sigma=(\sigma_v)_{v \in V}$, let $G(\sigma)$ denote the underlying undirected
graph whose node set is $V$, and whose edge set is
$E(\sigma)=\{(v,u): v \in V \wedge u \in \sigma_v\}$. Then, the \emph{cost} incurred by player
$v$ under $\sigma$ is

\begin{equation}
\label{cost} C_v(\sigma) = \alpha \cdot |\sigma_v| + \sum_{u \in V}
d_{G(\sigma)}(u,v)
\end{equation}

\noindent where $d_{G(\sigma)}(u,v)$ is the distance between nodes $u$
and $v$ in $G(\sigma)$. Thus, the cost function implements the inherently antagonistic
goals of a player, which on one hand
attempts to buy as little edges as possible, and on the other hand
aims to be as close as possible to the other nodes in the
resulting network. These two criteria are suitably balanced in
(\ref{cost}) by making use of the parameter $\alpha \geq 0$.
Consequently, the \emph{Nash Equilibria}\footnote{In this paper, we only
focus on \emph{pure} strategies Nash equilibria.} (NE) space of the game is
a function of it. Actually, if we characterize such a space in terms of the \emph{Price of Anarchy}
(PoA), then this has been shown to be constant for all values of $\alpha$ except for $n^{1-\varepsilon} \leq \alpha \leq 65 \, n$, for any $\varepsilon \geq 1/\log n$ (see \cite{MMM13,MS10}).

A first natural variant of \textsc{Sum}NCG was introduced in
\cite{DHM07}, where the authors redefined the player cost function
as follows

\begin{equation}
\label{max} C_v(\sigma) = \alpha \cdot |\sigma_v| + \max_{u \in V} d_{G(\sigma)}(u,v).
\end{equation}

\noindent This variant, named \textsc{Max}NCG, received further
attention in \cite{MS10}, where the authors improved the PoA of
the game on the whole range of values of $\alpha$, obtaining in this case that the PoA is constant for all values of $\alpha$ except for $129 > \alpha = \omega(1/\sqrt{n})$.

Besides these two basic models, many variations on the theme have been defined.
In an effort of defining $\alpha$-free models, in \cite{LPR08} the authors proposed a
variant in which a player, when forming the network, has a limited budget to establish links to other
players. This way, the player cost function will only describe the
goal of the player, namely either the maximum distance or the
total distance to other nodes. Since in \cite{LPR08} links and hence the resulting graph are seen as directed, a
natural variant of the model was given in \cite{EFM11}, where the undirected case was considered.
Afterwards, in \cite{BGP11} the authors proposed a model complementing the one given in \cite{EFM11}. More precisely, they assumed that the
cost function of each player now only consists of the number of bought edges (without any budget on them), and a player needs to connect to the network by satisfying the additional constraint of staying within a given either maximum or total distance to the rest of players.
Then, in \cite{ADH10} the authors proposed a further variant, called \textsc{Basic}NCG, in which given some
existing network, the only improving transformations allowed are
\emph{edge swap}, i.e., a player can only modify a \emph{single}
incident edge, by either replacing it with a new incident edge, or
by removing it. Recently, this model has been extended to the case in which edges are oriented and players can swap only outleading edges \cite{MS12}. Notice that, differently from the previous models, the problem of computing a best-response strategy of a player in \textsc{Basic}NCG is \emph{not} \np-hard. This inspired a subsequent model \cite{L12} in which a player can swap, add, or delete a single edge. Here the best response strategy is still computable in polynomial time while the player has more freedom to act.

Generally speaking, in all the above models the obtained results on the PoA are asymptotically worse than those we get in the two basic models, and we refer the reader to the cited papers for the actual bounds.

\paragraph{Our results.}

In this paper we concentrate on a seemingly underplayed generalization of NCGs, namely that in which for each player the set of possible adjacent nodes is constrained by a given connected, undirected graph $H$, which in the end will host the created network. This finds its practical motivations in the physical limitations of constructing links, and was originally studied in \cite{DHM09} for \textsc{Sum}NCG, where it is shown that the PoA is upper bounded by $O(1 + \sqrt{\alpha})$ and $O(1+\min\{\sqrt{n},n^2/\alpha\})$ for $\alpha < n$ and $\alpha \geq n$, respectively, and lower bounded by $\Omega(1 + \min\{\alpha/n,n^2/\alpha\})$. Here, we focus on the max-distance version, that we call \textsc{Max}NCG$(H)$, and we show that also in this case the PoA is fairly bad,\footnote{According to the spirit of the game, we concentrate on connected equilibria only. In fact, to avoid pathological disconnected equilibria, we can slightly modify the player's cost function (\ref{max}) as it was done in \cite{EFM11}, in order to incentivize the players to converge to connected equilibria. Alternatively, this can be obtained by assuming that initially the players sit on a connected network (embedded in the host graph), and they move (non-simultaneously) with a myopic best/improving response dynamics.} even when the host graph is restricted to some basic standard layout patterns. More precisely, we show that the PoA is $\Omega(1+\min\{\alpha, n/\alpha\})$ for the classes of $k$-regular (with any constant $k \geq 3$) and 2-dimensional grid host graphs. This lower bound is asymptotically tight for $\alpha=\Omega(\sqrt{n})$, since we can prove a general upper bound of $O\Big(1 + \frac{n}{\alpha+\rho_H}\Big)$, where $\rho_H$ is the radius of $H$. Moreover, on general host graphs, we exhibit a lower bound of $\Omega \Big(\sqrt{\frac{n}{1+\alpha}}\Big)$ for $0 \leq \alpha = o (n)$. Quite surprisingly, this implies that the PoA is $\Omega(\sqrt{n})$ even when the players can build edges for free. On the positive side, if $\alpha \geq n$, we show the PoA is at most 2 (this is a direct consequence of the fact that in this case any equilibrium is a tree). Finally, in the meaningful practical case in which the host graph is sparse (i.e., $|E(H)|=n-1+h$, with $h=O(n)$), we prove that the PoA is $O(1+h)$, and so for very sparse graphs, i.e. $h=O(1)$, we obtain that the PoA is constant.

Preliminarily to the above study, we also provide some results concerning the computational and dynamic aspects of the game.  First, we prove that computing a best response for a player is \np-hard for any $0<\alpha=o(n)$, thus extending a similar result given in \cite{MS10} for \textsc{Max}NCG when $\alpha=2/n$. Then, we prove that \textsc{Max}NCG$(H)$ is not a potential game, by showing that an improving response dynamic does not guarantee to converge to an equilibrium, even if we assume a minimal \emph{liveness} property that no player is prevented from moving for arbitrarily many steps. This implies that an improving response dynamic may not converge for the \textsc{Max}NCG game as well (after relaxing such a liveness property).
A deeper discussion on dynamics in NCGs can be found in \cite{KL13,L11}.

The paper is organized as follows: in Section 2 we analyze the computational/convergence aspects of the game, while Sections 3 and 4 discuss the upper and lower bounds to the PoA, respectively.

\section{Preliminary results}
First of all, we observe that, as for the sum-distance version of the problem studied in \cite{DHM09}, it is open to establish whether \textsc{Max}NCG$(H)$ always admits an equilibrium. This problem is particularly intriguing, since the topology of $H$ could play a discriminating role on that. We conjecture an  affirmative answer to this question for any $\alpha > 0$ (for $\alpha=0$ it is trivially true as any strategy profile $\sigma$ such that $G(\sigma)=H$ is a NE). As a first step towards this direction, observe that given any $H$, a breadth-first search tree rooted at a center of $H$, and in which each node owns the edges towards its children, is an equilibrium whenever $\alpha \geq \rho_H$, where $\rho_H$ is the radius of $H$. Indeed, notice that each vertex $u$ has no unactivated edges towards the vertices of its subtrees. This immediately implies that $u$ cannot improve its cost by changing to a strategy having at most the same number of bought edges. Moreover, in any other strategy, $u$ must incur a cost of at least $\alpha\geq \rho_H$ to buy some additional edges while its routing cost, which is actually less than or equal to $2\rho_H$ cannot decrease by more than $\rho_H$ since it can never be less than $\rho_H$.

Besides that, and similarly to other NCGs, we also have the bad news that the problem of computing the best response of player is \np-hard, as stated in the following theorem.

\begin{theorem}
\label{th:best response}
For every constant $\epsilon>0$ and for every $0<\alpha=o(n)$, the problem of computing a best response strategy of a player is \np-hard.
\end{theorem}
\begin{proof}
The reduction is from the \np-complete {\em 3-Exact 3-Cover} problem (3X3C for short) which, given (i) a set $O$ of $3k$ objects, with $k \in \mathbb{N}$, and (ii) a set ${\cal S}$ of subsets of $O$ each having cardinality equal to 3 and such that every $o\in O$ is a member of at most three sets in ${\cal S}$, asks for determining whether there exists a subset ${\cal S}'\subseteq {\cal S}$ of cardinality equal to $k$ that covers $O$, i.e., $\bigcup_{S \in {\cal S}'}S=O$. The reduction is the following. Let $\eta > (k+1)\alpha$. For a given instance of the 3X3C problem with $3k$ objects we build a host graph $H$ having $n=\Theta(\eta k)$ vertices. More precisely, $H$ contains an object vertex $v_o$ for every object $o \in O$, and a set vertex $u_S$ for every $S \in {\cal S}$. For every set $S \in {\cal S}$ and for every $o \in O$, $H$ contains a path of length $\eta$ having $v_o$ and $u_S$ as endpoints iff $o \in S$. All paths are vertex-disjoint except for at most their endpoints. Finally, $H$ contains an additional path $P$ of length $2\eta$ which is vertex-disjoint w.r.t. all the other paths, and an additional vertex $x$ linked to all the set vertices and to one endpoint of $P$, that we denote by $z$.

Let $\sigma$ be any strategy profile such that $G(\sigma)$ contains all edges of $H$ except those incident to $x$. We claim that any best response strategy of player $x$ has cost equal to $\alpha (k+1) +2\eta+1$ iff ${\cal S}$ contains a subset of cardinality $k$ that covers $O$.

To prove one direction, observe that if ${\cal S}'\subseteq {\cal S}$ has size $k$ and covers $O$, then, by buying the edge towards $z$ and all the $k$ edges $(x,u_S)$ for every $S \in {\cal S}'$, $x$ would incur a cost equal to $\alpha (k+1)+2\eta+1$. In fact, it is easy to see that each vertex of $G$ which is neither $x$ nor a vertex in $P$ is at distance of at most $\eta$ from some object vertex. Furthermore, each object vertex is at distance of at most $\eta$ from some set vertex $u_S$ such that $S \in {\cal S}'$.

To prove the other direction, let $\sigma'$ be the strategy profile obtained from $\sigma$ by changing $x$'s strategy with one of its best response strategies and assume that $C_{x}(\sigma')= \alpha (k+1)+2\eta+1$. Let ${\cal S}'$ be the subset of ${\cal S}$ containing set $S$ iff $x$ is buying the edge towards $u_S$ in $\sigma'$. First of all, observe that if ${\cal S}'$ does not cover $O$, or $x$ has not bought the edge towards $z$, then the routing cost of $x$ is greater than or equal to $3\eta+1>\alpha(k+1)+2\eta+1$. Consequently, $x$ has bought the edge towards $z$ and ${\cal S}'$ covers $O$. Therefore $|{\cal S}'|\geq k$. As the distance from $x$ to the endpoint of $P$ different from $z$ is $2\eta+1$, and since $C_x(\sigma')= \alpha(k+1)+2\eta+1$, it follows that $|{\cal S}'|=k$. \qed
\end{proof}

Notice that the result stated in Theorem \ref{th:best response} holds for any $0<\alpha=o(n)$, and so this extends the \np-hardness proof given in \cite{MS10} which holds for complete host graphs and $\alpha=2/n$.

We now discuss a negative result about the convergence of the improving response dynamics. To the best of our knowledge, this is the first result showing that an improving response dynamics on \textsc{MaxNCG} might not converge to an equilibrium. A deeper discussion about dynamics in NCGs can be found in \cite{KL13,L11}.

\begin{theorem}
	\label{thm:LB_PoA}
	For every value of $\alpha < \frac{n}{2} - 6$, \textsc{MaxNCG}$(H)$ is not a potential game. Moreover, if $\alpha>0$, an improving response dynamics may not converge to an equilibrium.
\end{theorem}
\begin{proof}
		We prove the non existence of a potential function by showing a cyclic sequence of strategy profiles where, at the end of each cycle, the total cost of the moving players has decreased by a positive constant.
			
	Let $l > \alpha+6$ be an integer satisfying $l \equiv 1 \pmod{3}$ and consider a host graph $H$ similar to the one shown in Figure \ref{fig:not_potential}. $H$ is composed by a cycle of $l$ nodes labelled from $1$  to $l$, by a path of $l-2$ edges having $x$ and $z$ as endpoints, and by all the edges between $x$ and the nodes of the cycle.
	The strategy profile $\sigma_a$ being played is shown by using a graphical notation explained in the caption.

\begin{figure}
	\center
	\includegraphics[scale=0.55]{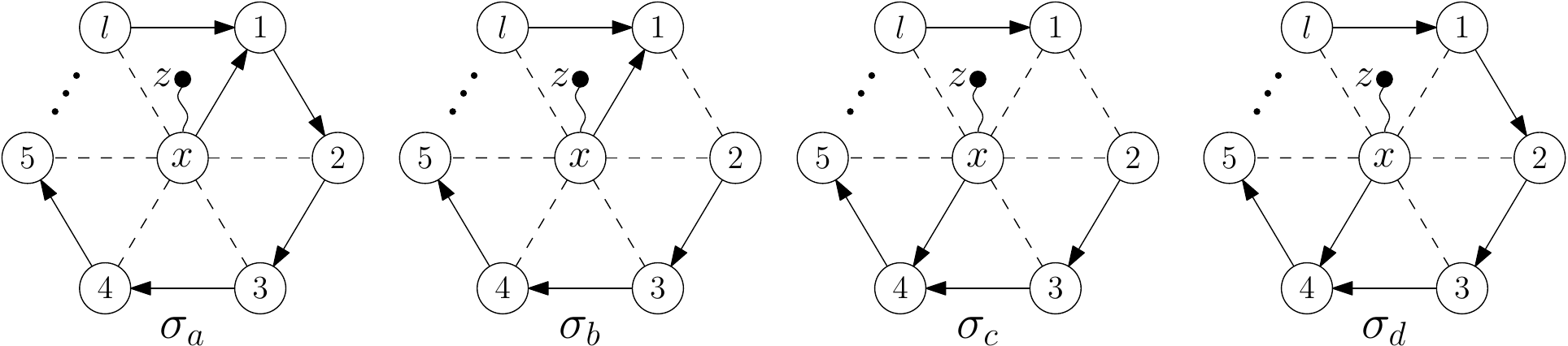}
	\caption{Representation of the strategy changes used in the proof of Theorem \ref{thm:LB_PoA}. On the left side, the initial configuration, where directed edges exit from their respective owner node, dashed edges are non-bought edges of $H$, and the spline denotes a path between $z$ and $x$, whose edges are arbitrarily owned.}
	\label{fig:not_potential}
\end{figure}

\noindent
In such a status, player $1$ is paying $\alpha+l-1$, whilst changing the strategy to $\sigma_b$ by removing the edge $(1,2)$ yields a cost of $l-1$, thus saving $\alpha$. Observe that now $C_x(\sigma_b)$ is $\alpha+l$, and so $x$ has interest in swapping the edge $(x,1)$ with the edge $(x,4)$, thus obtaining the strategy $\sigma_c$ and saving $1$.
	In such configuration $C_1(\sigma_c)$ has increased to $2l-4$, therefore player $1$ can buy back the edge $(1,2)$, as shown in strategy $\sigma_d$, thus reducing its cost to $\alpha+l+2$, i.e., saving $l-(\alpha+6)>0$.
	
	Notice how configuration $\sigma_d$ is similar to $\sigma_a$, with the only difference being the edge bought by $x$. Since $l \equiv 1 \pmod{3}$, by repeating $l$ times these strategy changes, every node in the cycle $\langle1, \dots, l\rangle$ will play a move at least once, and the resulting configuration is identical to $\sigma_a$, hence the players will cycle.
	
	To prove the latter part of the claim it suffices to note that after each cycle: (i) for $\alpha>0$ each strategy change is an improving response, and (ii) that the nodes in the path from $x$ to $z$ other than $x$, can never change their strategy.
	\qed
\end{proof}

Actually, the above proof shows that the improving response dynamics may not converge even if the minimal liveness property that each player takes a chance to make an improving move every fixed number of steps is guaranteed. Indeed, as observed in the proof, the players sitting on the path appended to $x$ do not move just because they cannot. Therefore, we can extend the above proof to the case in which $H$ is a complete graph if the players sitting on the path appended to $x$ can never move. This shows that the improving response dynamics may not converge on complete host graphs as well, i.e., for the classic NCG.

\section{Upper bounds}

In this section we prove some upper bounds to the PoA for the game. In what follows, for a generic graph $G$, we denote by $\rho_G$ and $\delta_G$ its radius and its diameter, respectively, and by $\varepsilon_G(v)$ the eccentricity of node $v$ in $G$. Moreover, we denote by $SC(\sigma)$ the \emph{social cost} of a generic strategy profile $\sigma$ (i.e., the sum of players' individual costs), and by $\mbox{\sc{Opt}}$ a strategy profile minimizing the social cost. Then

\begin{lemma}\label{lm:UB_aux}
Let $G=G(\sigma)$ be a NE, and let $\alpha=O(n)$. Then, we have that $SC(\sigma)/SC(\mbox{\sc{Opt}}) = O\big(\frac{\rho_G}{\alpha + \rho_H}\big)$.
\end{lemma}
\begin{proof}
Let $u$ be a center of $G$, and let $T$ be a shortest path tree of $G$ rooted at $u$. Clearly, the diameter of $T$ is at most $2 \, \rho_G$. Now, for every node $v$, let us denote by $k_v$ the number of edges of $T$ bought by $v$ in $\sigma$. The key argument is that if a node $v$ bought only the $k_v$ edges of $T$, its eccentricity would be at most $\varepsilon_T(v) \le 2 \, \rho_G$. Hence, since $\sigma$ is a NE, we have that $C_v(\sigma)\le \alpha k_v + 2 \, \rho_G$. By summing up the inequalities over all nodes, we obtain
$$
SC(\sigma)=\sum_v C_v(\sigma) \le \alpha \sum_{v} k_v + 2 n \, \rho_G= \alpha (n-1)+ 2n \, \rho_G.
$$
\noindent Now, since $SC(\mbox{\sc{Opt}}) \ge \alpha (n-1)+ n \, \rho_H$, we have
$$
\frac{SC(\sigma)}{SC(\mbox{\sc{Opt}})}\le \frac{\alpha (n-1)+ 2n \, \rho_G}{\alpha (n-1)+ n \, \rho_H}\le 1+ \frac{2 n \, \rho_G}{\alpha (n-1) + n \, \rho_H}=O\Big(\frac{\rho_G}{\alpha + \rho_H}\Big).
\displayqed
$$
\end{proof}

As an immediate consequence, we obtain the following:

\begin{theorem}
For $\alpha=O(n)$, the $PoA$ is $O(\frac{n}{\alpha+\rho_H})$.
\qed
\end{theorem}

Another interesting consequence of Lemma \ref{lm:UB_aux} concerns sparse host graphs:

\begin{theorem}\label{th:sparse host graph}
If the host graph $H$ has $n-1+h$ edges, and $h=O(n)$, then the $PoA$ is $O(h+1)$.
\end{theorem}
\begin{proof}
Let $G=G(\sigma)$ be an equilibrium network. Since $G$ must be connected, we have that $|E(H) \setminus E(G)| \le h$. This is sufficient to provide an upper bound to the diameter of $G$. Indeed, in \cite{SBvL87} it is shown that the diameter of a connected graph obtained from a supergraph of diameter $\delta$ by deleting $\ell$ edges is at most $(\ell+1)\delta$. This implies that in our case $\delta_G \le (1+h) \delta_H$. Now, the claim follows from Lemma \ref{lm:UB_aux}. \qed
\end{proof}

Theorem \ref{th:sparse host graph} implies that, for very sparse host graphs $H$, i.e., $|E(H)|=n-1+h$ and $h=O(1)$, we have that the PoA is $O(1)$, for any $\alpha$.

Next theorem shows that the PoA is upper bounded by 2 when $\alpha$ is at least $n$. A similar result has already been proved in \cite{DHM07} for the case in which $H$ is a complete graph. In fact, it turns out that the same proof extends to any host graph $H$. 

\begin{theorem}[\cite{DHM07}]
For $\alpha \ge n$, the $PoA$ is at most $2$.
\end{theorem}

We end this section by showing that when either $\alpha$ is small, or the host graph has small diameter, every stable tree (if any) is a good equilibrium. This generalizes a result in \cite{MS10} given for complete host graphs, which states that the social cost of every acyclic equilibrium is $O(1)$ times the optimum.

\begin{lemma}\label{lm:ecc(u)-ecc(v)<=1+a}
Let $(u,v) \in  E(H)$ be an edge of the host graph. Then, for every stable graph $G$, we have $|\varepsilon_G(u) - \varepsilon_G(v) | \le 1+\alpha$.
\end{lemma}
\begin{proof}
W.l.o.g. assume $\varepsilon_G(u) \ge \varepsilon_G(v)$. If $(u,v) \in E(G)$, then the claim trivially holds. Otherwise, if $u$ buys the edge $(u,v)$ then its eccentricity will decrease at least by $\varepsilon_G(u) - \varepsilon_G(v)-1$, while its building cost will increase by $\alpha$. Since $G$ is stable, we have $\varepsilon_G(u) - \varepsilon_G(v)-1 \le \alpha$, and the claim follows. \qed
\end{proof}

\begin{corollary}\label{crl:ecc(u)-ecc(v)<=1+a dH}
For every $u,v \in V$ and for every stable graph $G$, it holds that $| \varepsilon_G(u) - \varepsilon_G(v) | \le (1+\alpha) \, d_H(u,v)$.
\end{corollary}

\begin{lemma}\label{lm:ecc(u)-ecc(v) -> O(1+a)}
Let $G=G(\sigma)$ be a stable graph. If there are two nodes $u,v \in V$ such that $\varepsilon_G(v) \ge c\cdot \varepsilon_G(u) + k$ with $c > 1$ and $k\in \mathbb{R}$, then $\frac{\delta_G}{\delta_H} \le 2 \cdot \frac{1+\alpha - \frac{k}{\delta_H}}{c-1}$.
\end{lemma}
\begin{proof}
We have
\begin{eqnarray*}
		\varepsilon_G(v) - \varepsilon_G(u) & \ge & c\cdot \varepsilon_G(u) + k - \varepsilon_G(u) \\
		& \ge & (c-1) \varepsilon_G(u) + k \ge (c-1) \rho_G + k \ge  (c-1) \frac{1}{2}\delta_G + k.
	\end{eqnarray*}
	
	\noindent Moreover, from Corollary \ref{crl:ecc(u)-ecc(v)<=1+a dH}, we have
	\begin{eqnarray*}
		\varepsilon_G(v) - \varepsilon_G(u) & \le & (1+\alpha) d_H(u,v) \le (1+\alpha) \delta_H,
	\end{eqnarray*}
	
	\noindent from which, we obtain
	\begin{eqnarray*}
		(c-1) \frac{1}{2}\delta_G + k & \le & (1+\alpha) \delta_H
	\end{eqnarray*}
\hide{\\
		\frac{\delta_G}{\delta_H} & \le & 2 \cdot \frac{1+\alpha - \frac{k}{\delta_H}}{c-1}

}

\noindent and hence the claim. \qed
\end{proof}

We are now ready to give the following

\begin{theorem}
Let $\sigma$ be a NE such that $G=G(\sigma)$ is a tree. Then, $\frac{SC(\sigma)}{SC(\textsc{Opt})}\le \min \{O(1+\alpha),O(\rho_H)\}$.
\end{theorem}
\begin{proof}
Let us consider a center $u$ of $G$, and let $v$ be a node in the periphery of $G$, namely $\varepsilon_G(v)=\delta_G$. Since $G$ is a tree, we have $\varepsilon_G(v)=\delta_G \ge 2 \, \rho_G -1 = 2 \, \varepsilon_G(u) -1$. Now, using Lemma \ref{lm:ecc(u)-ecc(v) -> O(1+a)} and Lemma \ref{lm:UB_aux}, the claim follows. \qed
\end{proof}

\section{Lower bounds}

In this section we prove some lower bounds to the PoA of the game, as summarized in Table \ref{LB}.

\begin{table}
\begin{center}
\begin{tabular}{|l||c|c|c|}
\hline
{$\alpha$}& $O(\sqrt[3]{n})$ &$O(\sqrt{n})$ & $\Omega(\sqrt{n})$\\
\hline
PoA&$\Omega\left(\sqrt{\frac{n}{1+\alpha}}\right)$ & $\Omega(\alpha)$ &$\Omega\left(1+ \frac{n}{\alpha}\right)$\\
\hline
\end{tabular}
\label{LB}
\end{center}
\caption{Obtained lower bounds to the PoA.}
\end{table}

Before getting to the technical details, let us discuss the significance of the above bounds. First of all, we notice that the lower bound for $\alpha = \Omega(\sqrt{n})$ is tight, due the upper bound given in the previous section. Moreover, observe that we can obtain such a lower bound for two prominent classes of host graphs, namely for $k$-regular graphs (for any constant $k \geq 3$)  and for 2-dimensional grids.\footnote{Notice that a 2-dimensional grid is also planar and bipartite.} We view this as a meaningful result, due to the practical relevance of such host topologies.
Concerning the case $\alpha \in \Omega(\sqrt[3]{n}) \cap O(\sqrt{n})$, we notice that the lower bound still holds for the classes of $k$-regular graphs (for any constant $k\geq 3$) and 2-dimensional grids, but now it is not tight.  Finally, the lower bound of $\Omega(\sqrt{\frac{n}{1+\alpha}})$ for $\alpha = O(\sqrt[3]{n})$ is quite because it implies a very large lower bound of $\Omega(\sqrt{n})$ for tiny values of $\alpha$. Summarizing, we point out that we get a polynomial lower bound for $\alpha=O(n^{1-\varepsilon})$, for any $\varepsilon >0$, in strong contrast with the almost everywhere constant upper bound to the PoA of \textsc{Max}NCG.

\begin{theorem}
If the host graph is a 2-dimensional square grid, then the PoA is $\Omega\big(1+\min\{\alpha,\frac{n}{\alpha}\}\big)$.
\end{theorem}
\begin{proof}
Let $k = 2p$ where $p$ is an odd number, and let $H$ be a 2-dimensional square grid of $n=k\times k$ vertices. In the rest of the proof, we assume that the vertex in the $i$-th row and $j$-th column of the grid is labeled with $\langle i,j\rangle$, where $1\leq i,j\leq k$.

For every $1\leq j\leq k$, let $P_j$ be the path in $H$ which spans all the vertices of the $j$-th column of $H$. Let $k^*=\min\{ 1 + \lfloor \frac{\alpha}{2}\rfloor,k\}$. Let $F$ be the set of edges linking vertex $\langle 1, j\rangle$ with vertex $\langle 1, j+1\rangle$ iff $j$ is even and let $F'$ be the set of edges linking vertex $\langle k^*, j\rangle$ with $\langle k^*, j+1\rangle$ iff $j$ is odd.

\begin{figure}
\center
\includegraphics{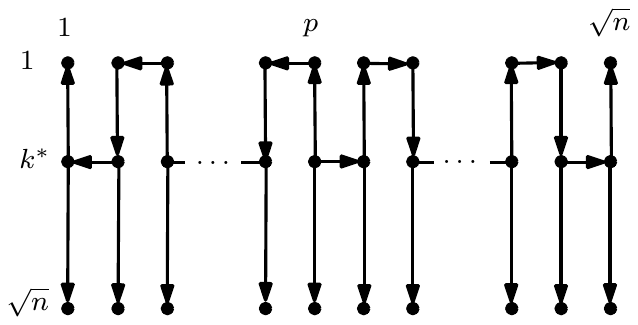}
\caption{The stable graph $G$ when the host graph $H$ is a square grid of $n$ vertices.}
\label{fig:grid}
\end{figure}

Let $G$ be the subgraph of $H$ whose edge set is $E(G)=F\cup F'\cup \bigcup_{j=1}^kE(P_j)$ (see also Figure \ref{fig:grid}). Observe that $G$ is a tree of radius greater than or equal to $\frac{1}{2}kk^*=\Omega\big(\sqrt{n} \cdot \min\{1+\alpha,\sqrt{n}\}\big)$. Observe also that $\langle k^*,p\rangle$ and $\langle k^*,p+1\rangle$ are the two centers of $G$. Let $\langle k^*,p\rangle$ be the root of $G$ and let $\bar{G}$ be the directed version of $G$ where all the root-to-leaf paths are directed towards the leaves. Finally, let $\sigma$ be the strategy profile induced by $\bar{G}$, i.e., each player $v$ is buying exactly the edges in $\bar{G}$ outgoing from $v$. Clearly, $G(\sigma)=G$.

We prove that $\sigma$ is a NE by showing that every vertex $\langle i,j\rangle$, with $1\leq i\leq k$ and $1\leq j\leq p$, is playing a best response strategy. Indeed, if we show that $\langle i,j\rangle$ is playing a best response strategy, then, by symmetry, also $\langle i, k-j+1 \rangle$ is playing a best response strategy.

Let $i$ and $j$ be two fixed integers such that $1\leq i\leq p$ and $1\leq j\leq k$ and let $t$ be the number of edges bought by $\langle i,j\rangle$ in $\sigma$.   Since $G$ is a tree and since $\langle k^*,p\rangle$ and $\langle k^*,p+1\rangle$ are the two centers of $G$, there exists a vertex $\langle i',j'\rangle$, with $1+p\leq i'\leq k$ and $1\leq j'\leq k$, such that the distance in $G$ from $\langle i,j\rangle$ to $\langle i',j'\rangle$ is exactly equal to the eccentricity of $\langle i,j\rangle$ in $G$. Observe also that the (unique) path in $G$ from $\langle i,j\rangle$ to $\langle i',j'\rangle$ traverses the root as well as the vertex $\langle k^*,p+1\rangle$. Let $\langle i',j'\rangle$ be any vertex such that $1+p\leq i'\leq k$ and $1\leq j'\leq k$. First of all, observe that if we add to $G$ all the edges adjacent to $\langle i,j\rangle$ in $H$, then the distance from $\langle i,j\rangle$ to $\langle i',j'\rangle$ decreases by at most $\alpha$. Since the cost of activating new links is at least $\alpha$, $\langle i,j\rangle$ cannot improve its cost function by buying more than $t$ edges. Now we prove that $\langle i,j\rangle$ cannot improve its cost function by buying at most $t$ edges. First of all, observe that $t$ is the minimum number of edges $\langle i,j\rangle$ has to buy to guarantee connectivity. Moreover, to guarantee connectivity, $\langle i,j\rangle$ has to buy an edge towards some vertex of every subtree of $G$ rooted at any of its $t$ children. Since the subtree of $G$ rooted at $\langle i,j\rangle$ does not contain $\langle i',j'\rangle$ when $\langle i,j\rangle$ is not the root, $\langle i,j\rangle$ cannot improve its eccentricity, and thus its cost function, by buying an edge towards some vertex of every subtree of $G$ rooted at any of its $t$ children. Furthermore, if $\langle i,j\rangle$ is the root of $G$, then $\langle i,j\rangle$ cannot improve its eccentricity, and thus its cost function, by buying an edge towards some vertex of every subtree of $G$ rooted at any of its $t$ children as $\langle i,j\rangle$ is already buying the unique edge of $H$ linking it to the subtree of $G$ rooted at $\langle k^*,1+p\rangle$.

To complete the proof, observe that $SC(\opt)$ is upper bounded by the social cost of $H$, i.e., $SC(\opt)= O\big(n(\alpha+\sqrt{n})\big)$. Since $SC(\sigma)\geq\alpha(n-1)+\frac{1}{2}kk^*n=\Omega\big(n^{3/2}\min\{1+\alpha,\sqrt{n}\}\big)$, we have that
$$
\frac{SC(\sigma)}{SC(\opt)}=\frac{\Omega\big(n^{3/2}\min\{1+\alpha,\sqrt{n}\}\big)}{O\big(n(\alpha+\sqrt{n})\big)}=
\Omega\big(1+\min\{\alpha,n/\alpha\}\big).
\displayqed
$$
\end{proof}

We now show that a similar lower bound holds also when the host graph is $k$-regular.

\begin{theorem}
	\label{thm:LB_PoA_k_regular}
	If the host graph is $k$-regular, with $3 \le k = o\left(\frac{n}{\alpha} \right)$, then the PoA is \linebreak $\Omega\big(1+\min\{\alpha,\frac{n}{\alpha k}\}\big)$.
\end{theorem}
\begin{proof}
	First of all, observe that for $\alpha=O(1)$ and $\alpha=\Omega(n)$ the claim trivially holds since the lower bound becomes $\Omega(1)$. Therefore, we consider the case $\alpha=\omega(1)$ and $\alpha=o(n)$. For the sake of readability, we provide the complete proof for even values of $k$ and we only sketch it for odd values of $k$ as the construction is very similar.

	Let $l$ be the greatest integer such that $l \le \lfloor \alpha + 1 \rfloor$, and let $\eta$ be a large value such that $\eta \equiv 1 \pmod{l}$. Notice that if the number of players $n$ is sufficiently large, then $l \ge 3$. We will use a host graph $H$ composed by: (i) a path $P$ of $\eta$ nodes, numbered from $0$ to $\eta-1$, (ii) a set of shortcut edges on $P$ (as described in the following), and (iii) a set of gadgets appended to $P$ and used to increase to $k$ the degree of its vertices (as described in the following).
	
	Concerning the shortcut edges, let $u_i$ be the node on $P$ numbered $i \cdot l$ for $i=0, \dots, g$, where $g= \frac{\eta-1}{l}$. Then, a shortcut edge connects $u_i$ to $u_{i+1}$, for $0\le i  < g$. Notice that any node on $P\setminus\{u_1, \ldots, u_{g-1}\}$ has now degree $2$ while the degree of all the nodes $u_1, \ldots, u_{g-1}$ is equal to 4.

	Concerning the gadgets, for each node $u$ on $P$ that has degree $d < k$, we augment $H$ as follows:
	\begin{itemize}
		\item we add a complete, loop-free, graph $K_u$ on $k+1$ vertices to $H$;
		\item we remove $\frac{k-d}{2}$ vertex-disjoint edges from $K_u$. So, $d$ vertices of $K_u$ have degree $k$ while the other $k-d$ vertices have degree $k-1$;
		\item we connect $u$ to the $k-d$ nodes of $K_u$ with degree $k-1$.
	\end{itemize}		
	At the end of this process the resulting host graph $H$ is $k$-regular.	Consider now a strategy profile $\sigma$ such that:
	\begin{itemize}
		\item all the edges of the path $P$ are bought (arbitrarily) by vertices other than $u_i$, $i=0, \ldots, g$;
		\item each vertex of a gadget that has an edge towards a node on $P$ buys it;
		\item the remaining vertices of the gadgets buy a single edge towards a vertex adjacent to a node of $P$.
	\end{itemize}		

	An example of the resulting configuration for $k=4$ along with the edges of the host graph is shown in Figure \ref{fig:LB_PoA_k_regular}(b).

\begin{figure}
	\center
	\includegraphics[scale=0.65]{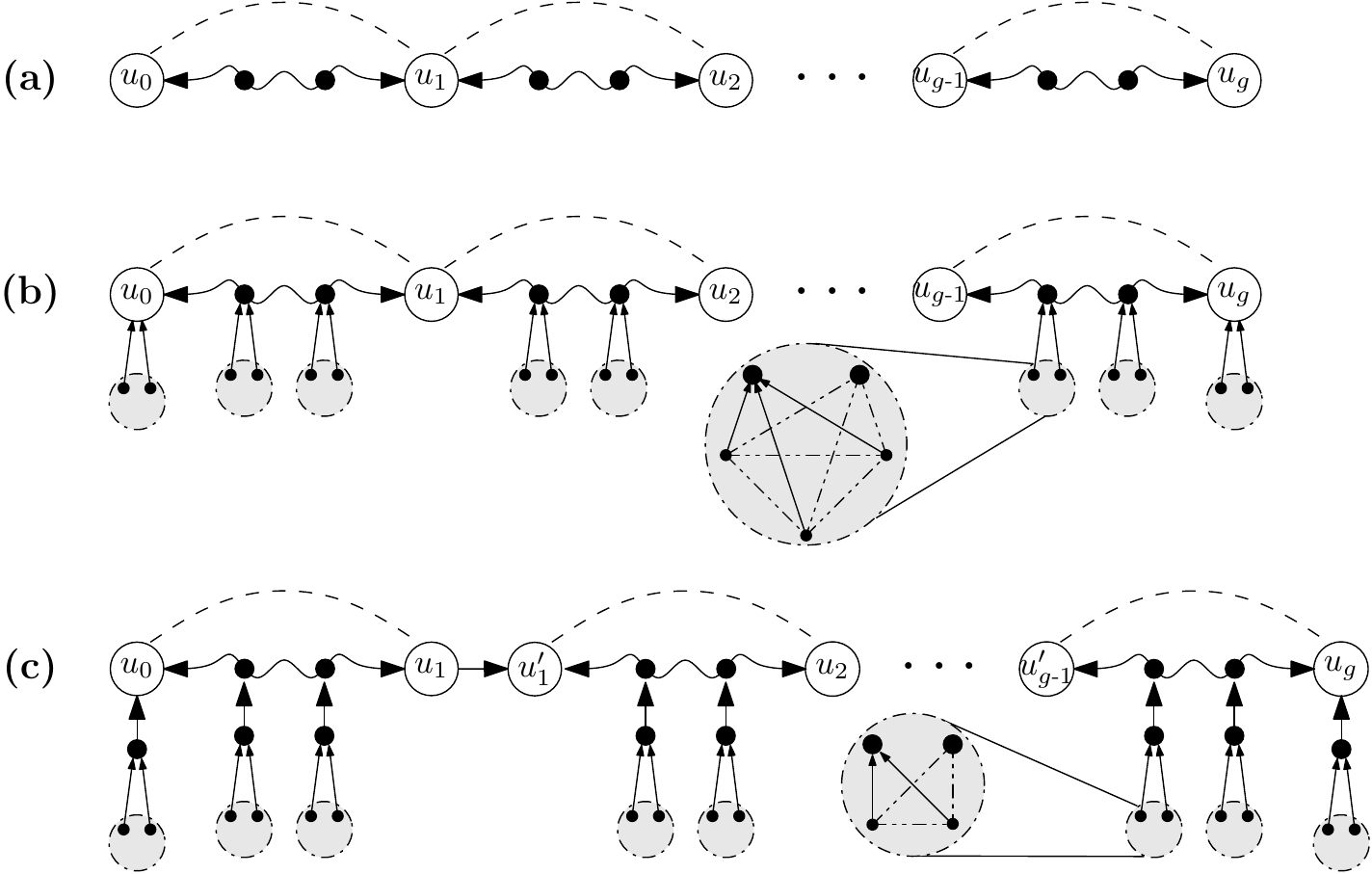}
	\caption{Representation of the host graph and the equilibrium used in the proof of Theorem \ref{thm:LB_PoA_k_regular} for (b) $k=4$, and (c) $k=3$.}
	\label{fig:LB_PoA_k_regular}
\end{figure}

	We show that $\sigma$ is a NE. Indeed, every node $u_i$ can only change its strategy by buying either one or two edges, but this can decrease its eccentricity by at most $l-1$, while increasing its building cost of at least $\alpha \ge l-1$. Moreover, the remaining nodes in $P$ cannot change their strategy, as doing so will cause the disconnection of the graph. Finally, the nodes of the gadgets buy just a single edge, and no other choice can decrease their eccentricity.
	
	Clearly $SC(\sigma)=\Omega(\alpha n + n\eta)$, as $G(\sigma)$ is a tree with radius $\Theta(\eta)$. Let now $\widehat{G}$ be the graph obtained by adding to $G(\sigma)$ the shortcut edges of $H$. The number of edges of $\widehat{G}$ is $n-1+g \le 2n$, and its diameter is bounded by $2 \cdot \varepsilon_{\widehat{G}}(u_0) \le 2 \cdot (g+l+2)$, as $u_0$ can take advantage of the shortcut edges. As a consequence, with a small abuse of notation, we have $SC(\widehat{G})=O(\alpha n + n(g+l))$.
	
	Using the relations $l=\Theta(\alpha)$, $\eta=\Theta(lg)$, and $n=\Theta(\eta k)$, we have that
	\begin{equation*}
		\mathrm{PoA} \ge \frac{SC(\sigma)}{SC(\widehat{G})} = \frac{\Omega(\alpha n + n \eta)}{O(\alpha n + n(g+l))} =
		\Omega\left(\frac{\eta}{\alpha+\frac{\eta}{\alpha}}\right)= \Omega\left(\frac{n}{\alpha k + \frac{n}{\alpha}}\right)
	\end{equation*}
	
\noindent
from which the claim easily follows.

If $k$ is odd, then a host graph similar to the one shown in Figure \ref{fig:LB_PoA_k_regular}(c) (for the case $k=3$) is considered. Notice that the shortcut edges are now vertex-disjoint, and each node incident to them has degree $3$, but for $u_0$ and $u_g$ that have degree 2. We first append a gadget to every node $x$ with degree $2$, in order to obtain a $3$-regular graph.
This gadget is a similar to a clique on $4$ vertices where an edge $e$ between two adjacent vertices has been replaced by two vertices going from the endpoints of $e$ to a new node $u$. The vertex $u$ is then connected to $x$ using a new edge (bought by $u$).
	If $k=3$ we are done, otherwise we can append a gadget similar to the one shown for the even case to each node of the graph we just constructed, in order to increase the degree of each node to $k$.	\qed
\end{proof}

Notice that we can extend the previous lower bound to outerplanar and series-parallel graphs just by considering the host graph composed by the path plus the shortcut edges, without any additional gadget, as shown in Figure \ref{fig:LB_PoA_k_regular}(a).

\begin{theorem}
The PoA is $\Omega\big(1+\min\{\alpha,\frac{n}{\alpha}\}\big)$, even when the host graph is a an outerplanar or a series-parallel graph.
\end{theorem}

We end this section by proving a non-constant lower bound to the PoA when $\alpha=o(n)$. Remarkably, our lower bound implies a non-constant lower bound to the PoA for every small value of $\alpha$, i.e., players buy edges for free. Our lower bounding construction is a non-trivial modification of the 2D-torus-rotated-45${}^\circ$ construction used in \cite{ADH10} to prove a lower bound for \textsc{Basic}NCG.

\begin{theorem}\label{lm:torus}
For $\alpha=o(n)$, the PoA is $\Omega\left(\sqrt{\frac{n}{1+\alpha}}\right)$.
\end{theorem}
\begin{proof}
Let $k\in\mathbb{N}$ and let $\bar H$ be an edge-weighted 2D-torus-rotated-45${}^\circ$ consisting of $2k^2$ vertices that we call junction vertices. For every pair of integers $0\leq i,j< 2k$, with $i+j$ even, there is exactly one vertex of $\bar H$ labeled with $\langle i,j\rangle$. We treat the two integers of a vertex label as modulo $2k$. Each vertex $\langle i,j\rangle$ has exactly four neighbors in $\bar H$: $\langle i\pm1,j\pm 1\rangle$. All edge weights are equal to $\ell=2(1+\lceil\alpha\rceil)$. For every pair of integers $0\leq i,j< 2k$, let $X_{i,j}=\{\langle i',j'\rangle\mid i'=i \text{ or }j'=j\}$. The properties satisfied by $\bar H$ are the following:

\begin{enumerate}

\item[(i)] $\bar H$ is vertex transitive, i.e., any vertex can be mapped to any other by a vertex automorphism, i.e., a relabeling of vertices that preserves edges;
\item[(ii)] the distance between two vertices $\langle i,j\rangle$ and $\langle i',j'\rangle$ in $\bar H$ is equal to \linebreak $\ell \cdot \max\big\{\bar d(i,i'),\bar d(j,j')\big\}$, where $\bar d(h,h')=\min\big\{|h-h'|,2k-|h-h'|\big\}$;
\item[(iii)] the eccentricity of each vertex in $\bar H$ is equal to $\ell k$;
\item[(iv)] for every $0\leq i,j<2k$, the distance from every vertex $v \in X_{i,j}$ to vertex $\langle |i-k|,|j-k|\rangle$ is equal to $\ell k$;
\item[(v)] for every edge $e$ of $\bar H$, the eccentricity of both endpoints of $e$ in $\bar H - e$ is greater than or equal to $\ell(k+1)$;
\item[(vi)] for every edge $e$ of $\bar H$ and for every vertex $\langle i,j\rangle$, the distance from $\langle i,j\rangle$ and the closest endpoint of $e$ is less than or equal to $\ell(k-1)$.

\end{enumerate}
It is easy to see that (i) holds and it is also easy to see that (iv) holds once (ii) has been proved. To prove (ii), it is enough to observe that each label can change by $\pm 1$ each time we move from one vertex to any of its neighbors. To prove (iii), we use (i) and the fact that the distance from vertex $\langle i',j'\rangle$ to $\langle k,k\rangle$, which is equal to $\max\{|k-i'|,|k-j'|\}$, is maximized for $i'=j'=0$. To prove (v), we first use (i) to assume that, w.l.o.g, $e$ is the edge linking $\langle k,k\rangle$ with $\langle k-1,k-1\rangle$. Next, we observe that any path in $\bar H-e$ going from $\langle k,k\rangle$ to $\langle 1,1\rangle$ must traverse a neighbor $v$ of $\langle k,k\rangle$ in $\bar H-e$ and the distance between $v$ and $\langle 1,1\rangle$ in $H$ is equal to $\ell k$ because one of the two integers in the label of $v$ is equal to $k+1$. Finally, to prove (vi), we first use (i) to assume that, w.l.o.g., $i=j=k$, i.e., $\langle i,j\rangle$ is $\langle k,k\rangle$, and the two endpoints of $e$ are, respectively, $\langle i',j'\rangle$ and $\langle i'+1,j'+1\rangle$, where $0\leq i',j'<k$. Using (ii), it is easy to see that the distance from $\langle k,k\rangle$ to $\langle i'+1,j'+1\rangle$ is less than or equal to $\ell (k-1)$.

Let $G$ be an unweighted graph obtained from $\bar H$ by replacing each edge of $\bar H$ with a path of length $\ell$ via the addition of $\ell-1$ new vertices per edge of $\bar H$. Let $H$ be the host graph obtained from $G$ by adding an edge between $\langle i,j\rangle$ and every vertex in $X_{i,j}$, for every junction vertex $\langle i,j\rangle$ (see also Figure \ref{fig:torus}). Notice that the number of vertices of $H$ is $n=2k^2+4k^2(\ell-1)=\Theta\big((1+\alpha) k^2\big)$. In what follows, we call the vertices in $H$ which are not in $\bar H$ path vertices.

\begin{figure}
\center
\includegraphics[scale=0.7]{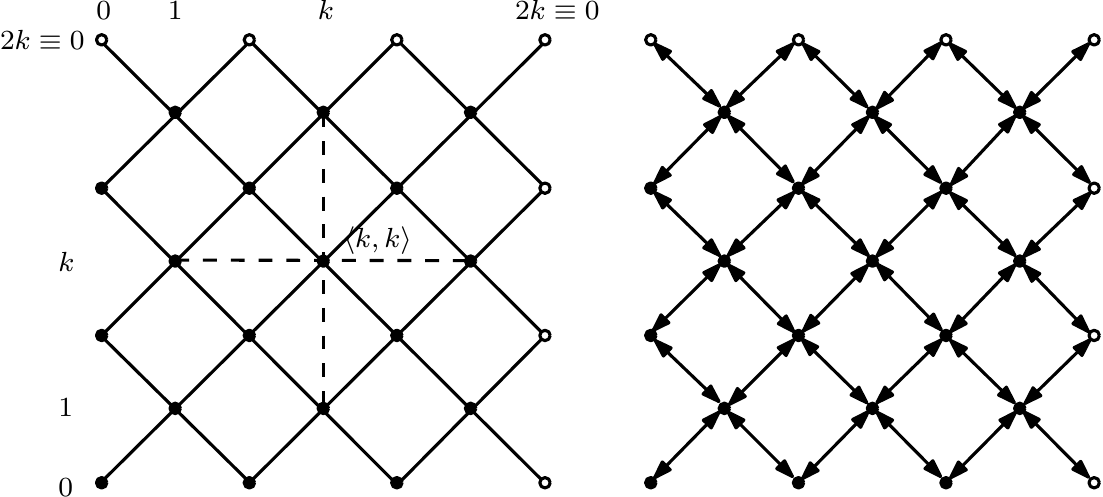}
\caption{The lower bound construction of Theorem \ref{lm:torus}. On the left side, the host graph $H$ is depicted. For the sake of readability, only junction vertices are visible and not all the edges are shown. The white vertices of row $2k\equiv 0$ are copies of the vertices of row $0$ while the white vertices of column $2k\equiv 0$ are copies of the vertices of column $0$. The solid edges are paths of length $\ell$, while the dashed edges are all the other edges adjacent to vertex $\langle k,k\rangle$. On the right side, the stable graph $G$ is depicted.}
\label{fig:torus}
\end{figure}

Let $\sigma$ be any strategy profile such that $G(\sigma)=G$ and all edges of $G(\sigma)$ are bought by players sitting on path vertices, i.e., no edge of $G(\sigma)$ is bought by some player sitting on junction vertices. We prove that $\sigma$ is a NE.

We start proving that players sitting on junction vertices are playing a best response strategy. Let $\langle i,j\rangle$ be a junction vertex. Observe that $\langle i,j\rangle$ is not buying any edge, therefore it suffices to show that $\langle i,j\rangle$ cannot improve its cost function by buying edges. First of all, observe that by (iii) and (vi), the eccentricity of $\langle i,j\rangle$ in $G$ is equal to $\ell k$. Indeed, if $v$ is a path vertex of a path $P$ corresponding to edge $e$ of $\bar H$, then the distance from $\langle i,j\rangle$ to the closest endpoint of $P$ (which corresponds to the closest endpoint of $e$) is less than or equal to $\ell(k-1)$. Therefore, the distance from $\langle i,j\rangle$ to $v$ is less than or equal to $\ell k$. To prove that $\langle i,j\rangle$ is in equilibrium, simply observe that if we add to $G$ all edges of $H$ incident to $\langle i,j\rangle$, i.e., all edges linking $\langle i,j\rangle$ to vertices in $X_{i,j}$, then by (iv) the distance from $\langle i,j\rangle$ to $\langle |i-k|,|j-k|\rangle$ is still $\ell k$.

Now, we prove that players sitting on path vertices are playing a best response strategy. Let $v$ be a path vertex. First of all, the eccentricity of $v$ in $G$ is less than or equal to $\ell k +\frac{1}{2}\ell$ by (vi). Indeed, if $v$ is a vertex of a path $P$ corresponding to edge $e$ of $\bar H$, then the distance from any junction vertex to the closest endpoint of $P$ (which corresponds to the closest endpoint of $e$) is less than or equal to $\ell(k-1)$. Therefore, the distance from $v$ to every junction vertex is less than or equal to $\ell k$ and the distance from $v$ to every other path vertex is less than or equal to $\ell k +\frac{1}{2}\ell$. Now, observe that $G$ already contains all edges of $H$ incident to $v$ and, moreover, the degree of $v$ in $G$ is equal to 2. Therefore, $v$ might improve its cost function by removing exactly one edge incident to it, i.e., by buying fewer edges than those it is buying in $\sigma$. However, if $v$ removes any of its incident edges in $G$, thus saving an $\alpha$ factor from its building cost, then by (v) the eccentricity of the unique junction vertex closest to $v$ becomes greater than or equal to $\ell(k+1)$ and thus, the eccentricity of $v$ also becomes greater than or equal to $\ell (k+1)$. Since $\ell(k+1)-\alpha\geq \ell (k+1) -\frac{\ell-2}{2}>\ell k + \frac{1}{2}\ell$, $v$ does not improve its cost function by buying fewer edges than those it is buying in $\sigma$.

To complete the proof, we have to show that PoA is $\Omega\left(\sqrt{\frac{n}{1+\alpha}}\right)$.
First of all, observe that the radius of $H$ is $\Theta(\ell)=\Theta(1+\alpha)$. Let $T$ be a breadth-first-search tree rooted at $\langle k,k\rangle$. Clearly the radius of $T$ is also $\Theta(1+\alpha)$. Furthermore, the social cost of $\opt$ is upper bounded by the social cost of $T$, i.e., $SC(\opt) \le \alpha(n-1)+n \cdot O(1+\alpha)=O\big((1+\alpha)n\big)$. As $SC(\sigma)\geq \alpha(4\ell k^2)+n\ell k=\Omega(n\ell k)=\Omega\big((1+\alpha) nk\big)$, we have that
$$
\frac{SC(\sigma)}{SC(\opt)}=\frac{\Omega\big((1+\alpha) nk\big)}{O\big((1+\alpha)n\big)}=\Omega(k)=\Omega\left(\sqrt{\frac{n}{1+\alpha}}\right).
\displayqed
$$
\end{proof}

\section{Conclusions}
In this paper, we studied the \textsc{MaxNCG} in the scenario in which the strategy space of all the players is constrained by a host graph. Our game is interesting for two reasons. First of all, it models the practical situation in which not all the links can be constructed due to physical limitations. Furthermore, the \np-hardness barrier of computing a best response strategy of a player in the \textsc{MaxNCG}, which clearly carries over into our game, can be easily broken if we restrict our game to the class of host graphs of (constant) bounded degree.

In our paper, we proved a strong lower bound of $\Omega(1+\min\{\alpha,n/\alpha\})$ to the PoA for two meaningful graph classes, namely bounded-degree host graphs and 2-dimensional square grids. Our lower bound asymptotically matches the general upper bound we provided for every $\alpha=\Omega(\sqrt{n})$. Since the PoA in the classical \textsc{MaxNCG} is mostly constant (see \cite{MS10}), we therefore have that the drawback of having a polynomial-time computability of best response strategies is the existence of stable networks of large social cost.

Finally, we concluded our paper by proving another lower bound of $\Omega(\sqrt{\frac{n}{1+\alpha}})$ to the PoA. Observe that all our lower bounds are never smaller than the lower bound of $\Omega(1+\min\{\alpha/n,n^2/\alpha\})$ known for the corresponding sum version of the game (see \cite{DHM09}). This is mainly due to the fact that reducing the routing cost of a player in our game is costlier, in terms of the building cost incurred by that player, than in the corresponding sum version of the game.

\end{document}